\documentclass
[journal,10pt,onecolumn,draftclsnofoot]{IEEEtran}
\usepackage[left=48 pt,right=48 pt,top=57 pt,bottom=43 pt]{geometry}

\newcommand{\phil}[1]{\textcolor{blue}{#1}}

\usepackage{graphicx}
\graphicspath{ {./images/} }

\usepackage{amsmath,amsfonts,amssymb,color}
\usepackage{cases}
\usepackage{epsfig}
\usepackage{psfrag}
\usepackage{algorithm}
\usepackage{algpseudocode}
\usepackage{array}
\usepackage{epstopdf}
\usepackage{tikz, subcaption, cite}

\usepackage{amsthm}
\newtheorem{assumption}{Assumption} 
\newtheorem{definition}{Definition} [section]
\newtheorem{theorem}{Theorem} 
\newtheorem{proposition}{Proposition} 
\newtheorem{lemma}{Lemma} [section] 

\newtheorem*{remark}{Remark}

\begin{document}
%
\title{
    The Impact of Vaccine Hesitancy on Epidemic Spreading}
%
%
%

\author{C.~H.~Leung,~Mar\'ia~E.~Gibbs,~and~Philip~E.~Par\'e*
    \thanks{*C. H. Leung and P. E. Par\'e are with the School of Electrical and Computer Engineering, Purdue University. M. E. Gibbs is with the Department of Aerospace Engineering and Engineering Mechanics, The University of Texas at Austin. This work was supported in part by the National Science Foundation, grants NSF-CNS \#2028738 and NSF-ECCS \#2032258.}
}
\maketitle

\begin{abstract}
    The COVID-19 pandemic has devastated the world in an unprecedented way, causing enormous loss of life. Time and again, public health authorities have urged people to become vaccinated to protect themselves and mitigate the spread of the disease. However, vaccine hesitancy has stalled vaccination levels in the United States. This study explores the effect of vaccine hesitancy on the spread of disease by introducing an SIRS-V$_\kappa$ model, with compartments of susceptible (S), infected (I), recovered (R), and vaccinated (V). We leverage the concept of carrying capacity to account for vaccine hesitancy by defining a vaccine confidence level $\kappa$, which is the maximum number of people that will become vaccinated during the course of a disease. The inverse of vaccine confidence is vaccine hesitance, $(\frac{1}{\kappa})$. We explore the equilibria of the SIRS-V$_\kappa$ model and their stability, and illustrate the impact of vaccine hesitance on epidemic spread analytically and via simulations.


\end{abstract}


%
\IEEEpeerreviewmaketitle


\section{Introduction}


%
%
%
%


The origins of epidemiological modeling extend as far back as the 18th century \cite{bookModelingInfectiousDiseases}, with Bernoulli's 1760 treatise on smallpox considered one of its foundational documents~\cite{BernoulliRevisited}. In the 20th century, the field experienced significant growth, including the introduction of approaches that employed dynamic systems theory \cite{bookModelingInfectiousDiseases}. In particular, Kermack and McKendrick's 1932 \textit{Contributions to the Mathematical Theory of Epidemics} introduced compartmental models, which are heavily used to this day \cite{kermack1932contributions,PARE2020Overview}. With this method, members of the population are compartmentalized based on their current state of health. The most common models are the SIS, SIR, and SEIR models, in which people are grouped as: \textit{susceptible (S), exposed (E), infected (I),} or \textit{recovered (R)} \cite{PARE2020Overview}. Expansions of the aforementioned models have also been developed. For instance, in light of the COVID-19 pandemic, Giordano \textit{et al.} constructed a SIDARTHE-V model that subdivides the infected population according to the detection and severity of their symptoms and includes vaccination levels \cite{
    giordano2021sidarthe-vaccines}.

We include a vaccinated compartment in our model because vaccines are a powerful tool in mitigating epidemic spread and saving lives \cite{bookModelingInfectiousDiseases}. The modern practice of vaccination is rooted in Edward Jenner's 1796 experiments on smallpox \cite{jenner1800inquirysmall_cow_pox,hilleman2000historyofvaccines}. Since then, incredible breakthroughs have been made in vaccinations, including the use of vaccines to eradicate smallpox \cite{henderson2011eradicationofsmallpox}. Today, vaccines are widely used around the world as early as infancy to prevent infections \cite{WHO_ImmunizationCoverage}. Due to the power of vaccines, the drastic death rates caused by the spread of the novel SARS-CoV-2 virus have spurred worldwide vaccination efforts to combat the virus \cite{dong2020JohnHopkinsDataInteractiveDashboard,OurWorldinData_vaccines2021global}. In the U.S., three vaccines are currently available: Pfizer-BioNTech, Moderna, and Johnson \& Johnson’s Janssen \cite{CDC_vaccine_types}.

Nevertheless, vaccine hesitancy remains a real, growing concern for public health \cite{dube2013vaccine_Hesitancyoverview}. An individual's vaccine hesitancy is affected by context, interpersonal experiences, and disease-specific metrics \cite{SAGEGroup2015vaccinehesitancydef}. In the United States, vaccine hesitancy is especially apparent regarding the reception of the COVID-19 vaccine. As of September 11, 2021, 53\% of Americans have been fully vaccinated against COVID-19, with an additional 9.2\% partially vaccinated \cite{OurWorldinData_vaccines2021global}. However, these levels are lower than those predicted by a February 2021 survey by Pew Research Center. According to the survey, 69\% of Americans reported they would probably, definitely, or had already been at least partially vaccinated \cite{Pew_AmericanVaccineOpinions}. In fact, the number of vaccines administered daily in the U.S. in September 2021 is less than half of those administered daily in mid-April 2021 \cite{OurWorldinData_vaccines2021global}.

There have been dozens of studies on COVID-19 vaccine hesitance around the world \cite{sowa2021VaccHesitancyPoland,oliveira2021Covid-19VaccineHesitancyBrazil,aw2021VaccineHes_High_Income_Countries}. There are also studies on how vaccination levels affect the dynamics of disease spread. Pires and Crokidakis explore the effect of vaccine opinions in the spread of a disease \cite{pires2017Epidemic_dynamicsw_Vaccination}. Similarly, \cite{pires2018sudden_ContinousOpinions} considers the dynamics when opinions are continuous. However, to the best of the authors' knowledge, the correlation between a maximum vaccination capacity $\kappa$ and the equilibria of system dynamics \cite{dhondt1988carrying}
has not yet been established. This work analytically illustrates the effect of $\kappa$ on the dynamics of COVID-19 spread. Specifically, $\kappa$ is shown to determine the endemic and disease-free equilibria of the SIRS-V$_\kappa$ system.

\section{SIRS-V$_{\kappa}$ Model}


In this paper, we take an SIRS model and add a vaccinated compartment with a vaccine confidence $\kappa$ as the upper-bound of vaccination level. The result is an \textit{SIRS-V$_\kappa$} model: \textit{susceptible (S), infected (I), recovered (R), vaccinated (V)}.
Figure \ref{fig:SIHRVDcompartments} gives a graphical representation of the compartmental flow.
\begin{definition}[Vaccine Hesitance]\label{def:vaccine_hesitance}
    We define the vaccine hesitance as $\kappa^{-1}$, where $\kappa$ is the vaccine confidence.
\end{definition}

Each compartment in the grouped SIRS-V$_\kappa$ model is a scalar between $0$ and $1$ modeling the fraction of the population in each epidemic compartment at time $t$.
The following are the equations describing the flow from one compartment to another in the group SIRS-V$_\kappa$ model:

\begin{subequations}\label{eq:sirsv_group}
    \begin{align}
        \dot{S} & = -\beta S I - \rho \left( 1 - \frac{V}{\kappa}\right) S + \omega R \label{eq:sirsv_group_S} \\
        \dot{I} & = \beta S I -\gamma I \label{eq:sirsv_group_I}                                               \\
        \dot{R} & = \gamma I - \omega R -\rho\left(1 - \frac{V}{\kappa}\right) R \label{eq:sirsv_group_R}      \\
        \dot{V} & = \rho\left(1 - \frac{V}{\kappa}\right) (S + R). \label{eq:sirsv_group_V}
    \end{align}
\end{subequations}
The parameters are defined as follows:
\begin{itemize}
    \item $\beta_{}$ is the frequency-dependent transmission rate
    \item $\gamma_{}$ is the recovery rate
    \item $\rho_{}$ is the rate of vaccination roll-out
    \item $\kappa$ is the vaccine confidence of the population
    \item $\omega$ is the rate of waning natural immunity\phil{.}
\end{itemize}
All of the above parameters are defined in the range of $(0, \infty)$.
\begin{figure}
    \centering
    \vspace{1ex}
    \includegraphics[width=\columnwidth]{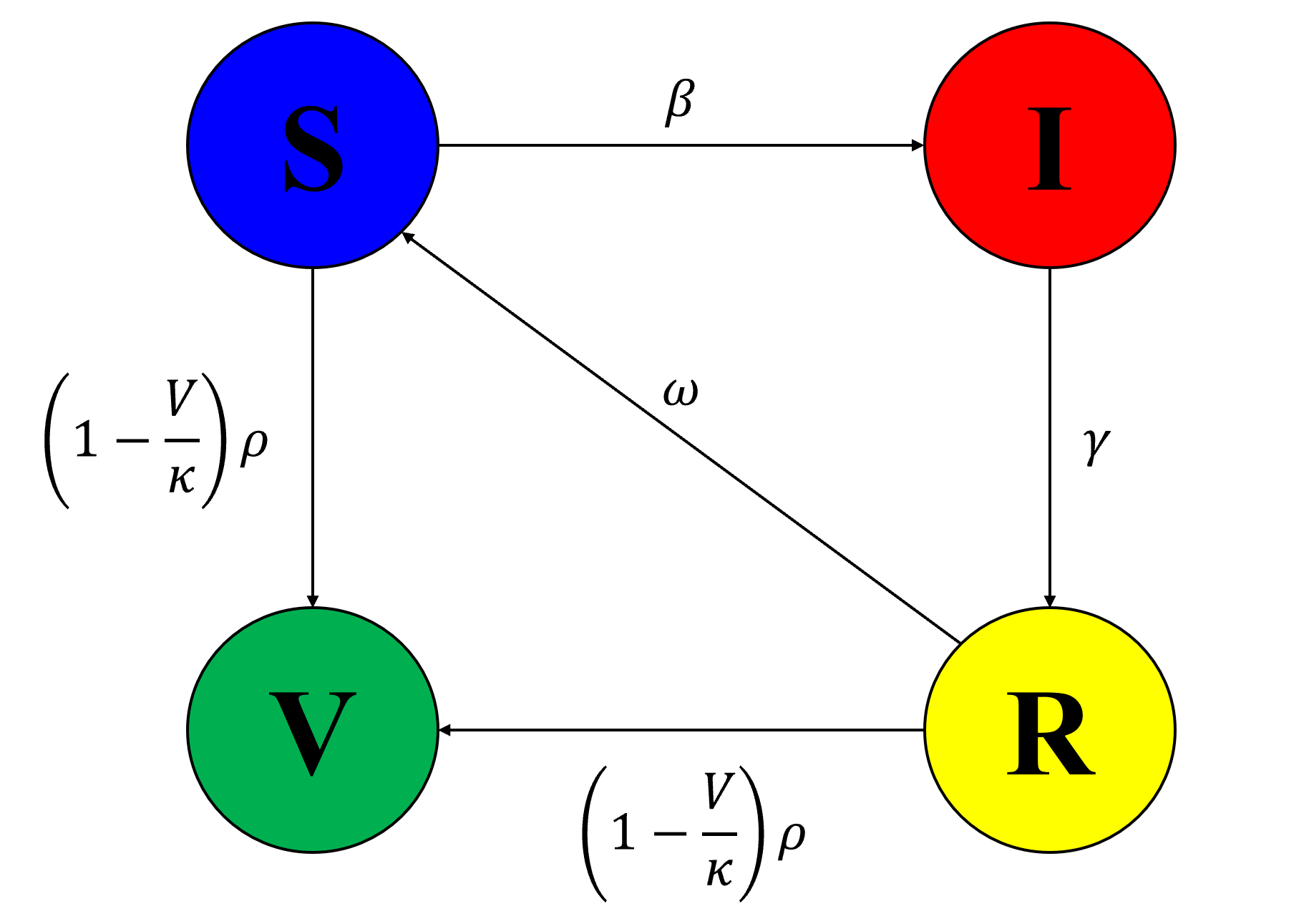}
    \caption{SIRS-V$_\kappa$ Compartments.}
    \label{fig:SIHRVDcompartments}
\end{figure}

\begin{assumption}[Frequency dependent and population preserving system]\label{asm:00}
    We assume $S(t_0), I(t_0), R(t_0)\in [0, 1]$, $V(t_0) \in [0, \underline{\kappa}]$, where $\underline{\kappa} := \min\{1, \kappa\}$, and $S(t_0) + I(t_0) + R(t_0) + V(t_0) = 1$ $\forall i \in \{1, ..., n\}$. Furthermore, all parameters of the system are strictly positive.
\end{assumption}
\begin{lemma}[Permissible range of $V$]\label{lem:kappa_range}
    Let Assumption~\ref{asm:00} hold, then
    $0 \leq V(t) \leq \underline{\kappa}\ \ \forall t \geq t_0, \forall i \in \{1, ..., n\}$.
\end{lemma}

\begin{proof}
    Suppose Assumption \ref{asm:00} holds,
    we first note that $\dot{V}(t_0) = \rho(S(t_0) + R(t_0)) \geq 0$ when $V(t_0) = 0$. Furthermore, $\lim_{V(t_0)\downarrow 0} \dot{V}(t_0) = \rho(S(t_0) + R(t_0)) \geq 0$.
    Then, consider the upper bound of $V(t)$, $\dot{V}(t_0) = 0$ when $V(t_0) = \underline{\kappa}$. Moreover, for any $V(t_0) \leq \underline{ \kappa}$, $\lim_{V(t_0)\uparrow \kappa} \dot{V}(t_0) \geq 0$. Therefore, $v(t) \in [0, \underline{\kappa}]$ $\forall t \geq t_0$ for all permissible initial values in Assumption \ref{asm:00}.
\end{proof}

The following lemma
shows that the bounds on the initial state apply to all states for all $t\geq t_0$.

\begin{lemma}[Admissible range of states values] \label{lem:00}
    Let Assumption~\ref{asm:00} hold. Then
    $S(t), I(t), R(t), V(t) \in [0, 1]$ and $S(t) + I(t) + R(t) + V(t) = 1$ $\forall t \geq 0, \forall i \in \{1, ..., n\}$.
\end{lemma}
\begin{proof}
    We first observe that $\dot{S}(t) + \dot{I}(t) + \dot{R}(t) + \dot{V}(t) = 0$ $\forall t \geq 0$
    by summing over \eqref{eq:sirsv_group_S}-\eqref{eq:sirsv_group_V}, then by taking the integration:
    \begin{align*}
        \int_{t_0}^t \dot{S}(\tau) + \dot{I}(\tau) + \dot{R}(\tau) + \dot{V}(\tau) d\tau & = 0                     \\
        S(t) + I(t) + R(t) + V(t)                                                        & = 1\ \ \forall t \geq 0
    \end{align*}
    since the initial condition is given as $S(t_0) + I(t_0) + R(t_0) + V(t_0) = 1$.

    To show that $S(t), I(t), R(t), V(t) \in [0, 1]$ $\forall t \geq 0$, we rewrite the state vector as $z(t) \in \mathbb{R}^{4}$, then
    %
    %
    we note that for any $j \in \{1, \hdots, 4\}$, $\dot{z}_j = -z_jf_j(z_{k\neq j}) + g_j(z)$ where $f_j(z_{k\neq j}), g_j(z) \geq 0$ $\forall z \geq 0$, also $z_{k\neq j}$ is the state vector $z$ without its $j^{th}$ entry.
    This applies to $\dot{z}_j = \dot{V}$ also, as we have demonstrated in Lemma~\ref{lem:kappa_range}.
    Then, we observe that $\lim_{z_j\downarrow 0}\dot{z}_j(t) = g_j(z) \geq 0\ \ \forall z \geq 0$. Therefore, $z(t) \geq 0$ when $z(t_0) \geq 0$.
    Then $z_j(t) \leq 1$ follows directly from  $z_j(t) \geq 0$ and $S(t) + I(t) + R(t) + V(t) = 1$.
\end{proof}

We denote the set of admissible states of \eqref{eq:sirsv_group} by $\mathcal{S} = \{(S(t), I(t), R(t), V(t)):\ S(t) + I(t) + R(t) + V(t) = 1\ \wedge\ S(t), I(t), R(t) \in [0, 1]\ \wedge\ V(t) \in [0, \underline{\kappa}]\ \forall t \geq t_0\}$.

\section{
  Stability Analysis}
This section will lay down the existence and uniqueness of the two equilibria of the SIRS-V$_\kappa$ model and their stability conditions.
We denote an equilibrium of the SIRS-V$_\kappa$ model as $(S^*,I^*,R^*,V^*)$ where $S^*$, $I^*$, $R^*$, and $V^*$ are the steady-state values of \eqref{eq:sirsv_group} as $t \to \infty$.
\begin{definition}[Disease-Free Equilibrium]\label{def:DFE}
    A disease-free equilibrium (DFE) is an equilibrium with the steady-state infected level $I^{(d)} = 0$ $\forall i \in \{1,...,n\}
    $, where $\dot{S}^{(d)} = \dot{I}^{(d)} = \dot{R}^{(d)} = \dot{V}^{(d)} = 0$.
\end{definition}
Similarly, we define the endemic equilibrium as follows:
\begin{definition}[Endemic Equilibrium Point]\label{def:eep}
    An endemic equilibrium point (EEP) is an equilibrium with the steady-state infected level $I^{(e)} > 0$,
    where $\dot{S}^{(e)} = \dot{I}^{(e)} = \dot{R}^{(e)} = \dot{V}^{(e)} = 0$.
\end{definition}
\begin{lemma}[Strictly Positive EEP]\label{lemma:positive_eep}
    If there exists an endemic equilibrium $(S^{(e)}, I^{(e)}, R^{(e)}, V^{(e)})$, then it is strictly greater than $0$, i.e. $z^{(e)} > 0$ $\forall z \in \{S, I, R, V\}$.
    Furthermore, $V^{(e)} = \underline{\kappa}$.
\end{lemma}
\begin{proof}
    Assume, by way of contradiction,
    that $S^{(e)} = 0$,
    then $\dot{S}^{(e)} = -\gamma I \neq 0$, which contradicts the definition of EEP. Therefore, $S^{(e)} > 0$.
    A similar argument can be made by substituting $R^{(e)} = 0$ into \eqref{eq:sirsv_group_R} to get $R^{(e)} > 0$.
    Therefore, from \eqref{eq:sirsv_group_V},
    we have
    \begin{equation}\label{eq:eep_sirs_network_v}
        0 = \rho \left(1 - \frac{V^{(e)}}{\kappa}\right) (S^{(e)} + R^{(e)})
    \end{equation}
    and the only way for the R.H.S. to equal zero is
    if $V^{(e)} = \kappa$. Notice that \eqref{eq:eep_sirs_network_v} cannot be satisfied when $\kappa \geq 1$, and therefore, EEP only exists when $\kappa < 1$,
    concluding the proof.
\end{proof}
Lastly, we reiterate the definition of the basic reproduction number through the next generation matrix method
\cite{diekmann1990definition, van2002reproduction}:
\begin{definition}[Basic and Effective Reproduction Number]
    Let $\dot{z_i} = f_i(z) - g_i(z)$ be the differential equation of the $i^{th}$ infected compartment, where $f_i(z)$ is the function governing the rate of appearance, $g_i(z)$ is the function governing the rate of transferring into other compartments, and $i \in \{1, \hdots, m\}$ for $m$ infected compartments out of $n$ total compartments. Let $F=\left[\frac{\partial f_i(z_0)}{\partial z_j}\right]$ and $G=\left[\frac{\partial g_i(z_0)}{\partial z_j}\right]$ be the Jacobians of $[f_i]$ and $[g_i]$, $\forall j \in \{1, \hdots, m\}$, then $FV^{-1}$ is the next generation matrix, and its spectral radius $\sigma(FV^{-1})$ is the basic reproduction number $\mathcal{R}_0$ of the system. $z_0$ is the initial condition where full susceptibility of population is assumed ($S(t_0) \approx 1$).
    The effective reproduction number is defined as $\mathcal{R}_t :=S(t)\mathcal{R}_0$.
\end{definition}
The $ij^{th}$ entry of $FV^{-1}$ is the expected number of secondary cases in the $i^{th}$ compartment produced by an infected individual in the $j^{th}$ compartment \cite{van2017reproduction}. In the case of a single infected compartment, it is simply the average number of secondary infected cases introduced by one infected individual.

\begin{proposition}
    The basic reproduction number $\mathcal{R}_0$ of \eqref{eq:sirsv_group} is $\frac{\beta}{\gamma}$, and the effective reproduction number $\mathcal{R}_t$ is upper bounded by $(1 - V(t))\frac{\beta}{\gamma}$.
\end{proposition}
\begin{proof}
    From \eqref{eq:sirsv_group_I}, $f(I) = \beta SI$, $
        g
        (I) = \gamma I$. By computing the Jacobians, we get $F =\frac{\partial\beta SI(t_0)}{\partial I} = \beta
        S(t_0)
    $ and $G =\frac{\partial\gamma I}{\partial I}= \gamma$. Therefore, $\mathcal{R}_0 = \sigma(FG^{-1}) = \frac{\beta}{\gamma}$ since $\mathcal{R}_0$ is defined for $S(t_0) \approx 1$. Further, $\mathcal{R}_t = S(t)\frac{\beta}{\gamma} = (1 - V(t) -R(t))\frac{\beta}{\gamma} \leq (1 - V(t))\frac{\beta}{\gamma}$.
\end{proof}


Notice that we use the word uniqueness in the following text not to refer to the uniqueness of the set of all equilibria of \eqref{eq:sirsv_group} but
of a particular class of equilibria, i.e. the DSF and EEP are sets of cardinality at most one.
As we will see
in the
following,
the EEP could be an empty set under certain conditions.

\begin{proposition}[Existence and uniqueness of DFE] \label{prop:dfe}
    There always exists a unique disease-free equilibrium $(S^{(d)}, I^{(d)}, R^{(d)}, V^{(d)})$ = $(1 - \underline{\kappa}, 0, 0, \underline{\kappa})$.
\end{proposition}

\begin{proof}
    By the definition of a disease-free equilibrium and from \eqref{eq:sirsv_group_I}, we know the $j^{th}$ DFE is characterized by $I^{(d_j)}=0$, and $\dot{S}^{(d_j)}$ $=$ $\dot{I}^{(d_j)}$ $=$ $\dot{R}^{(d_j)} $ $=$ $\dot{V}^{(d_j)}$ $=$ $0$ $\forall j \in \mathbb{R}$. Equation~\eqref{eq:sirsv_group_R} can be evaluated as
    $
        0 = \left(\omega + \rho\left(1 - \frac{V^{(d_j)}}{\kappa}\right)\right)R^{(d_j)}
    $
    by substitution, which implies either:
    \begin{numcases}{}
        V^{(d_1)} = \kappa\left(1 + \frac{\omega}{\rho}\right) & or \label{eq:grouped_dfe_case1}\\
        R^{(d_2)} = 0. \label{eq:grouped_dfe_case2}
    \end{numcases}
    In the first case \eqref{eq:grouped_dfe_case1}, we can substitute $V^{(d_1)}$ into \eqref{eq:sirsv_group_V} and solve for $S^{(d_1)} + R^{(d_1)} = 0$. By Assumption \ref{asm:00} and
    Lemma \ref{lem:00}, $S^{(d_1)} = R^{(d_1)} = 0$ is the only solution to the equation. Together with the presumptions $V^{(d_1)} = \kappa\left(1 + \frac{\omega}{\rho}\right)$ and $I^{(d_1)} = 0$, we realize that $S^{(d_1)} + I^{(d_1)} + R^{(d_1)} + V^{(d_1)} \neq 1$,
    which is not in the permissible domain of states by Assumption~\ref{asm:00} and Lemma~\ref{lem:00}. Therefore, DFE $(d_1)$ is not a feasible equilibrium. 

    Now, by \eqref{eq:grouped_dfe_case2}, we can substitute $R^{(d_2)} = 0$ into \eqref{eq:sirsv_group_S}, then we have $0 = \rho\left(1 - \frac{V^{(d_2)}}{\kappa}\right)S^{(d_2)}$, which leads to two subcases:
    \addtocounter{equation}{-1}
    \begin{subequations}
        \begin{numcases}{}
            S^{(d_{21})} = 0 & or \label{eq:grouped_dfe_case2a}\\
            V^{(d_{22})} = \kappa. \label{eq:grouped_dfe_case2b}
        \end{numcases}
    \end{subequations}
    %
    %
    In the first subcase \eqref{eq:grouped_dfe_case2a},
    since $I^{(d_{21})} = 0$ by the definition of the DFE and $R^{(d_{21})} = 0$ by \eqref{eq:grouped_dfe_case2},
    we
    have $(S^{(d_{21})}, I^{(d_{21})}, R^{(d_{21})}, V^{(d_{21})})$ $=$ $(0, 0, 0, 1)$, which is true when $\kappa \geq 1$
    by recalling Lemma \ref{lem:kappa_range} and that $\underline{\kappa} = \min\{1, \kappa\}$.
    The second subcase \eqref{eq:grouped_dfe_case2b}, combined with $I^{(d_{22})} = 0$, \eqref{eq:grouped_dfe_case2}, and
    Lemma \ref{lem:00},
    gives
    $(S^{(d_{22})}, I^{(d_{22})}, R^{(d_{22})}, V^{(d_{22})}) = (1 - \kappa, 0, 0, \kappa)$ which is true when $\kappa \in (0, 1]$.
    DFE $(d_{21})$ and $(d_{22})$ can be combined as $(1 - \underline{\kappa}, 0, 0, \underline{\kappa})$ to be applicable to the whole range of $\kappa \in (0, \infty)$.
    Since $(1 - \underline{\kappa}, 0, 0, \underline{\kappa})$ is the only DFE satisfying Definition~\ref{def:DFE} and \eqref{eq:sirsv_group}, it is unique under Assumption~\ref{asm:00}. Since $(1 - \underline{\kappa}, 0, 0, \underline{\kappa})$ is a valid state under Assumption~\ref{asm:00} for all admissible values of $\kappa$,
    the DFE always exists.
\end{proof}

\begin{remark}
    Note that $\kappa$ ceases to act as an upper bound on
    the vaccination
    state
    while still affecting the effective rate of vaccination $\rho \left(1 - \frac{V}{\kappa}\right)$ when $\kappa$ is greater than $1$. Moreover, $\rho \left(1 - \frac{V}{\kappa}\right) \to \rho$ as $\kappa \to \infty$.
\end{remark}

\begin{proposition}[Uniqueness of EEP] \label{prop:eep}
    The endemic equilibrium point of \eqref{eq:sirsv_group}:
    \small
    \begin{equation*}
        \left(S^{(e)}, I^{(e)}, R^{(e)}, V^{(e)}\right) = \left(\frac{\gamma}{\beta}, \frac{1 - \kappa - \gamma / \beta}{1 + \gamma/\omega}, \frac{1 - \kappa - \gamma / \beta}{1 + \omega/\gamma}, {\kappa}\right)
    \end{equation*}

    \normalsize

    \vspace{-1.5ex}

    \noindent
    is unique.
    Moreover,
    $$\frac{I^{(e)}}{R^{(e)}} = \frac{\omega}{\gamma}.$$
\end{proposition}
\begin{proof}
    If an endemic equilibrium exists, by Definition \ref{def:eep},
    $I^{(e)}>0$. Applying this fact
    and setting
    \eqref{eq:sirsv_group_I}
    to zero gives
    $S^{(e)}=\frac{\gamma}{\beta}$.
    Therefore, the result from
    Lemma \ref{lemma:positive_eep},
    with $n=1$, implies
    $V^{(e)}=\kappa$. By substituting $V^{(e)}=\kappa$ into \eqref{eq:sirsv_group_R}, we will have $ \frac{I^{(e)}}{R^{(e)}} = \frac{\omega}{\gamma}$.
    To solve for $I^{(e)}$ and  $R^{(e)}$, recall the global preservation of population result in Lemma~\ref{lem:00}, we have $I^{(e)} + R^{(e)} = 1 - \kappa - \frac{\gamma}{\beta}$. We can combine this equation with $\frac{I^{(e)}}{R^{(e)}} = \frac{\omega}{\gamma}$, then solving for the system of two equations with two unknowns gives us the desired result. Therefore, we have shown the uniqueness of permissible EEP under Assumption~\ref{asm:00}.
\end{proof}
Notice that since $I^*$ can either be $0$ or $(0, 1]$, the DFE and the EEP are the only two possible
admissible equilibria of
\eqref{eq:sirsv_group} under
Assumption~\ref{asm:00}.

\begin{lemma}[Necessary and sufficient condition for existence of EEP]\label{lem:nns_eep_exists}
    The endemic equilibrium exists in \eqref{eq:sirsv_group} if and only if $\kappa < 1 - \frac{\gamma}{\beta}$, which is equivalent to $(1 - \kappa)\frac{\beta}{\gamma} > 1$.
\end{lemma}
\begin{proof}
    From Assumption~\ref{asm:00} and Lemma~\ref{lem:00},
    the EEP in Proposition~\ref{prop:eep} exists if and only if
    the range of parameters satisfy the four inequality:
    \begin{equation}\label{eq:eep_param_range}
        \begin{cases}
            0 < \frac{\gamma}{\beta}< 1                                     \\
            0 < \frac{1 - \kappa - \gamma/\beta}{1 + \gamma / \omega} < 1   \\
            0 < \frac{1 - \kappa - \gamma / \beta}{1 + \omega / \gamma} < 1 \\
            0 < \kappa < 1.
        \end{cases}
    \end{equation}
    The first and forth constraints are absorbed by Assumption \ref{asm:00} and the second constraint. The second and third constraints can be reduced to $0 < 1 - \kappa - \gamma / \beta$, which is equivalent to $\kappa < 1 - \gamma/\beta$ and $(1 - \kappa)\frac{\beta}{\gamma} > 1$.
\end{proof}

Fig.~\ref{fig:state_space_with_EEP} illustrates the set of states $\mathcal{S}$, which is a 3-dimensional simplex with six vertices $(S, I, R, D, F, G)$ with coordinates $(1, 0, 0 ,0)$, $(0, 1, 0 ,0)$, $(0, 0, 1 ,0)$, $(1-\kappa, 1, 0 ,\kappa)$, $(0, 1-\kappa, 0 ,\kappa)$, and $(0, 0, 1-\kappa,\kappa)$, respectively.
The larger simplex is a tetrahedron because of Lemma~\ref{lem:00}, which states that $S(t) + I(t) + R(t) + V(t) = 1$ and $(S(t), I(t), R(t), V(t))\in[0,1]^4$. To see this clearly, consider an SIR compartmental model with three states, $(S, I, R)$. If $S + I + R = 1$, then $(S, I, R)$ lie only on the surface of the unit circle defined by the 1-norm. Furthermore, if $(S, I, R) \in [0, 1]^3$, then $(S, I, R)$ lies only on the surface belonging to the first quadrant of the 3-dimensional space, where $S$, $I$, $R$ are non-negative. Notice that the surface is an equilateral triangle with its three vertices $S = (1, 0, 0)$, $I = (0, 1, 0)$, and $R = (0, 0, 1)$. Generalizing this notion to the 4-dimensional case gives us Fig.~\ref{fig:state_space_with_EEP}.


\begin{theorem}[DFE GAS Conditions]\label{thm:DFEGAS_condtion}
    The DFE of
    \eqref{eq:sirsv_group} is globally asymptotically stable (GAS)
    if and only if $(1 - \kappa)\frac{\beta}{\gamma} \leq 1$.
\end{theorem}
\begin{proof}
    Our goal is to show by exhaustion of cases that all the states in $\mathcal{S}$
    are
    either the DFE or will approach the DFE as $t \to \infty$ under the condition $(1 - \kappa)\frac{\beta}{\gamma} \leq 1$. Fig.~\ref{fig:state_space_with_EEP} visualizes the simplex of states, i.e. the range of admissible states $\mathcal{S}$, of which we want to show GAS around the DFE.

    Suppose $(1 - \kappa)\frac{\beta}{\gamma}\leq 1$. 
    If $S(t)+R(t) = 0$, which is equivalent to $S(t)=R(t)=0$, then either $(S(t)=R(t)=I(t)=0 \wedge V(t)=1)$, which is a special case of DFE when $\kappa \geq 1$, or $(S(t)=R(t)=0 \wedge I(t)>0)$. If $(S(t)=R(t)=0 \wedge I(t)>0)$, then by \eqref{eq:sirsv_group_R}:
    \begin{align*}
        \dot{R}(t) & =
        \gamma I(t) > 0.
    \end{align*}
    %
    %
    Hence, every admissible state in \eqref{eq:sirsv_group} which satisfies $S(t)+R(t) = 0$ will have $\dot{R}(t) > 0$, hence evolving to another state such that $S(t)+R(t) \neq 0$. In other words, states on the $IV$ edge in Fig.~\ref{fig:state_space_with_EEP} are not stable and will move toward vertex $R$ as time unfolds. The only state on $IV$ which violates this tendency of evolving towards $R$ is the vertex $V$, which is already the DFE if it lies within $\mathcal{S}$.
    For every state that satisfies $(S(t) + R(t) > 0 \wedge V(t) < \kappa)$,
    \begin{align*}
        \dot{V}(t) & = \rho\left(1 - \frac{V(t)}{\kappa}\right)(S(t)+R(t)) \\
                   & > 0
    \end{align*}
    %
    %
    by substituting $S(t) + R(t) > 0$ into \eqref{eq:sirsv_group_V}.
    Therefore, every state with $V(t) < \kappa$ will satisfy $V(t) = \kappa$ as $t \to \infty$ or it is already the DFE.

    If $(I(t) = R(t) = 0 \wedge V(t) = \kappa)$, then $S(t) = 1 - \kappa$ by Assumption~\ref{asm:00} and Lemma~\ref{lem:00}, which state that $S(t) + I(t) + R(t) + V(t) = 1$ $\forall t \geq t_0$,
    which is the DFE.
    Otherwise, $(V(t) = \kappa \wedge I(t)
        + R(t)
        > 0)$. Thus, by substituting $V(t) = \kappa$ into $S(t)+I(t)+R(t)+V(t)=1$, we have
    \begin{align}
        S(t) & = 1 - \kappa - I(t) - R(t)\nonumber                      \\
             & \leq \frac{\gamma}{\beta} - I(t) - R(t) \label{eq:ineq1} \\
             & < \frac{\gamma}{\beta}, \label{eq:ineq2}
    \end{align}
    %
    %
    where \eqref{eq:ineq1} holds because the condition $(1 - \kappa)\frac{\beta}{\gamma}\leq 1$ is equivalent to $\frac{\gamma}{\beta} \geq 1 - \kappa$ and \eqref{eq:ineq2} holds because $I(t)
        + R(t)
        > 0$.
    Therefore,
    from \eqref{eq:sirsv_group_I} and \eqref{eq:ineq2}, we have
    \begin{align*}
        \ \ \ \dot{I}(t) & = \beta \frac{\gamma}{\beta}I(t) - \gamma I(t) \\
                         & = -2\gamma I(t)                                \\
                         & \leq 0,
    \end{align*}
    %
    %
    where equality holds only if $I(t)>0$.
    %
    %
    Therefore, every state which satisfies $V(t) = \kappa$ will satisfy $(V(t) = \kappa \wedge I(t) = 0)$ as $t\to \infty$ or it is already the DFE.

    Lastly, for any state that satisfies $(V(t) = \kappa \wedge I(t) = 0 \wedge R(t) > 0)$, from \eqref{eq:sirsv_group_R}, we have
    \begin{align*}
        \dot{R}(t)
        = - \omega R(t)
        < 0.
    \end{align*}
    %
    %
    Thus, every state that satisfies $(V(t) = \kappa \wedge I(t) = 0 \wedge R(t) > 0)$ will satisfy $(V(t) = \kappa \wedge I(t) = 0 \wedge R(t) = 0)$
    as $t \to \infty$ or it is already the DFE. Therefore, the DFE is GAS if $(1 - \kappa)\frac{\beta}{\gamma} \leq 1$.

    If $(1 - \kappa)\frac{\beta}{\gamma} > 1$, then there exists
    an EEP in $\mathcal{S}$ by Lemma~\ref{lem:nns_eep_exists}, which concludes the
    other direction of the
    proof.
\end{proof}

\begin{figure}
    \centering
    \includegraphics[width=\columnwidth]{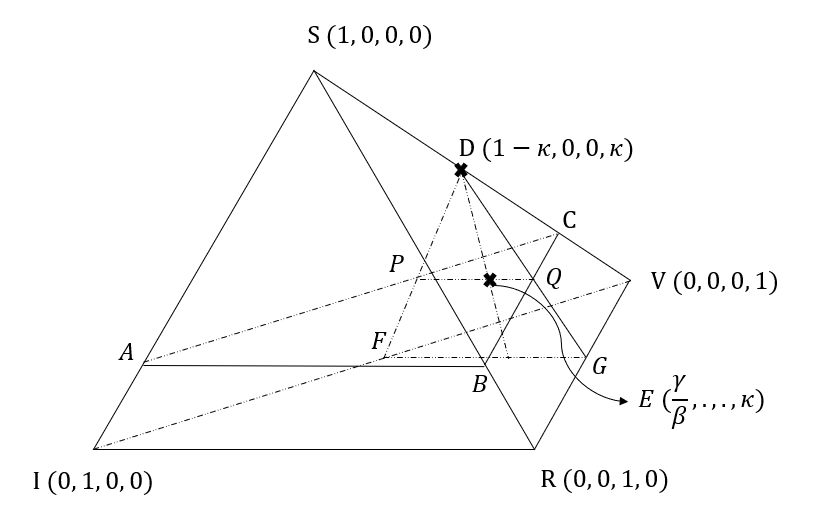}
    \caption{Admissible range of state values of \eqref{eq:sirsv_group} when EEP exists. $D$ is the DFE, and $E$ is the EEP. The
        plane $ABC$ represents all states with $S(t)=\frac{\gamma}{\beta}$, and the
        plane $DFG$ represents all states with $V(t)=\kappa$. The ratio  between the segments $\frac{\|EQ\|}{\|PE\|} = \frac{\omega}{\gamma}$.}
    \label{fig:state_space_with_EEP}
\end{figure}

Referring to the points in Fig.~\ref{fig:state_space_with_EEP}, $(1 - \kappa)\frac{\beta}{\gamma} \leq 1$ is achieved when the
plane
$ABD$ slides up or $DFG$ slides back, such that the segment $PQ$ vanishes. One can show that all states above the $ABC$
plane have $\dot{I} > 1$ and below $\dot{I} < 1$. An intuitive way to interpret Theorem~\ref{thm:DFEGAS_condtion} is that when $ABC$ is strictly above $DFG$, all states are pushed away from $I$ and pulled toward $V$, until they reach the DFE.

\begin{remark}
    Notice that the necessary and sufficient GAS condition of the disease-free equilibrium is not $\mathcal{R}_0 < 1$ for
    the SIRS-V$_\kappa$
    model, because the basic reproduction number only guarantees sufficient local asymptotic stability when it is less than one. Therefore, it might at times over/under-estimate the threshold condition of the spread process.
\end{remark}

\section{Simulations}
In this section, we investigate the impact of the vaccine confidence $\kappa$ on the behavior of the SIRS-V$_{\kappa}$ model. Fig.~\ref{fig:kappa_v} compares the trajectories of the vaccination level $V(t)$ with different values of $\kappa$, ranging from $0.1$ to $\infty$.
Note that, consistent with the analysis results, $\kappa$ slows the rate of convergence and, if $\kappa<1$, can lower the limit of $V(t)$. Fig.~\ref{fig:kappa_v} and Fig.~\ref{fig:kappa_I} were plotted with initial condition $(S(t_0), I(t_0), R(t_0), V(t_0)) = (0.54,0.41,0.05,0)$. Fig.~\ref{fig:models_compare} and Fig.~\ref{fig:peak_inf_peak_time} were plotted with initial condition $(S(t_0), I(t_0), R(t_0), V(t_0)) = (0.99,0.01,0,0)$, and the initial condition of Fig.~\ref{fig:threshold} is $(S(t_0), I(t_0), R(t_0), V(t_0)) = (0.7,0.3,0,0)$. All plots use the parameter set: $(\beta = 1.6, \gamma = .8, \rho = .12, \kappa = .8, \omega = .2)$ except for Fig.~\ref{fig:kappa_I} and Fig.~\ref{fig:kappa_v} we have used $\omega = 3$.

\begin{figure}
    \centering
    \includegraphics[width=\columnwidth]{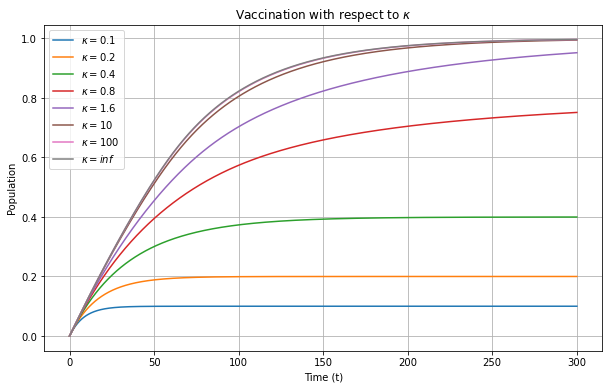}
    \caption{The impact of varying $\kappa$ on $V(t)$}
    \label{fig:kappa_v}
\end{figure}

Fig.~\ref{fig:kappa_I} illustrates how varying $\kappa$ affects the infection $I(t)$. Since $\frac{\beta}{\gamma} = 2$ in this case,
$\kappa = 0.5$ barely satisfies the condition for GAS in Theorem~\ref{thm:DFEGAS_condtion} to reach the DFE as time goes to infinity. We also notice that $\kappa$ slows the rate of convergence of $I(t)$ when comparing the cases $\kappa = (0.5, 0.6, 1.2, \infty)$.

\begin{figure}
    \centering
    \includegraphics[width=\columnwidth]{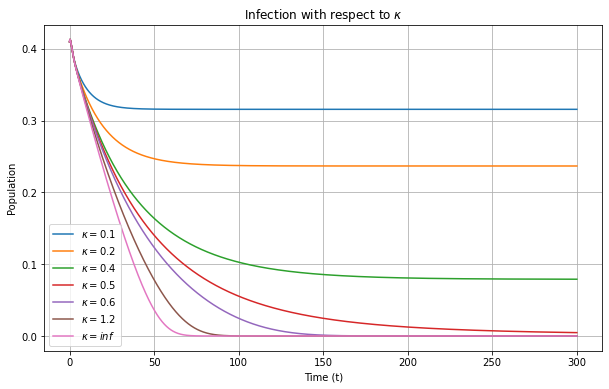}
    \caption{The impact of varying $\kappa$ on $I(t)$}
    \label{fig:kappa_I}
\end{figure}

Fig.~\ref{fig:models_compare} compares $I(t)$ and $R(t)$ between the SIRS, SIRSV, and SIRS-V$_{\kappa}$ models. Notice the key difference between the three models is that SIRSV and SIRS-V$_{\kappa}$ has $\rho > 0$ while $SIRS$ has $\rho = 0$, and SIRSV$_{\kappa}$ has $\kappa = 0.8 < \infty$ while SIRSV assumes $\kappa \to \infty$. In this particular setup, both the SIRSV and SIRSV$_{\kappa}$ models reach the DFE before $t = 40$, while the SIRS model settles at the EEP. Furthermore, we observe a higher infection peak with slower peak time in the SIRS-V$\kappa$ model.

\begin{figure}
    \centering
    \includegraphics[width=\columnwidth]{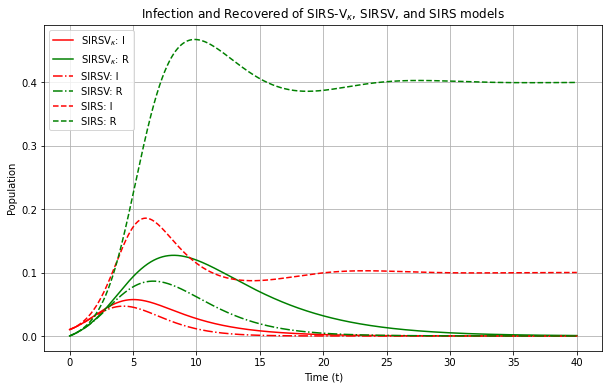}
    \caption{Comparation between SIRS, SIRSV, and SIRS-V$_{\kappa}$} models.
    \label{fig:models_compare}
\end{figure}

To investigate how
the vaccine confidence
affects the maximum peak infection value and time, we plot the maximum peak infection value and time versus $\kappa$ in Fig.~\ref{fig:peak_inf_peak_time}. Note that the maximum peak infection value decreases monotonically with $\kappa$. On the other hand, the peak infection time reaches its maximum when $\kappa = 0.2$ before decaying, which is not expected and will need further investigation.

\begin{figure}
    \centering
    \includegraphics[width=\columnwidth]{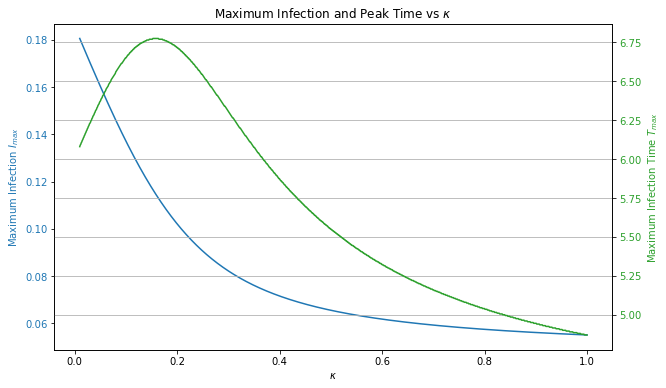}
    \caption{Maximum Infection value and peak time versus $\kappa$.}
    \label{fig:peak_inf_peak_time}
\end{figure}

Fig.~\ref{fig:threshold} plots the earliest time when $I(t) < 0.001$ with respect to different $(1 - \kappa)\frac{\beta}{\gamma}$ through varying $\kappa$. We use the condition $I(t) < 0.001$ as a way to investigate the asymptotic behavior of the system under different parameter settings. The program returns $-10$ if $I(t)$ fails to converge close enough to $0$ before $t = 100$. The sharp drop at $(1 - \kappa)\frac{\beta}{\gamma} \approx 1$ aligns with our finding in Lemma~\ref{lem:nns_eep_exists} and Theorem~\ref{thm:DFEGAS_condtion} that the DFE ceases to be GAS when $(1 - \kappa)\frac{\beta}{\gamma}$ is larger than $1$. Note that the drop occurs slightly before $1$ because we do not allow $t>100$, which is much smaller than asymptotic time. Finally, all the simulations indicate that the EEP
has
a large region of convergence as long as it exists in $\mathcal{S}$. More investigation is required to firmly conclude these findings analytically.

\begin{figure}[h]
    \centering
    \includegraphics[width=\columnwidth]{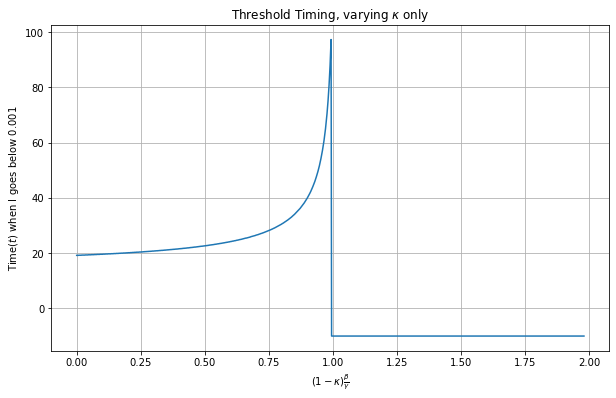}
    \caption{Threshold behavior with respect to $(1 - \kappa)\frac{\beta}{\gamma}$.}
    \label{fig:threshold}
\end{figure}

\section{Conclusion and Future Work}
In this paper, we have proposed an SIRS-V$_\kappa$ model,
where $\kappa$ is the vaccine confidence level. We have proved the existence of unique endemic and disease-free states, both of which depend on $\kappa$.
Furthermore, $\kappa$ acts as an important component in determining the necessary and sufficient condition of the global asymptotic stability, namely $(1 - \kappa)\frac{\beta}{\gamma} \leq 1$, of the disease-free equilibrium. From the perspective of control, manipulating the transmission rate $\beta$, recovery rate $\gamma$, and the vaccine confidence $\kappa$ are equally viable ways of mitigating epidemic spread. From the disease modeling perspective, we have shown through analytical and numerical methods that ignoring vaccine confidence in the model will introduce significant biases when determining the system's stability condition, maximum peak infection time, and its threshold behaviors.
The dependence of COVID-19 prevalence on vaccine hesitance should encourage Americans to become vaccinated if they have not done so already.

For future work, we are interested in extending our findings from the scalar case to the networked case to further inspect the impact of varying vaccine confidence across different geographic/demographic populations. Moreover, we also would like to provide analysis on
the region of convergence of EEP and DFE when $(1 - \kappa)\frac{\beta}{\gamma} > 1$, and the effect of vaccine confidence on the maximum peaking time of $I(t)$.

\bibliographystyle{IEEEtran}
\bibliography{bibrefs}









\end{document}